\documentclass{article}
\usepackage{spconf,amsmath,epsfig}
\usepackage{macros}
\usepackage{amsthm}
\usepackage[retainorgcmds]{IEEEtrantools}

\newtheorem{theorem}{Theorem}
\newtheorem{lemma}[theorem]{Lemma}
\newtheorem{corollary}[theorem]{Corollary}

\title{MULTIPATH CHANNELS OF UNBOUNDED CAPACITY}
\name{Tobias Koch ~~~~ Amos Lapidoth}
\address{ETH Zurich, Switzerland\\Sternwartstrasse 7, CH-8092
  Zurich\\Email: \{tkoch, lapidoth\}@isi.ee.ethz.ch}
\begin{document}
%
\maketitle
\begin{abstract}
The capacity of discrete-time, noncoherent, multipath fading channels
is considered. It is shown that if the variances of the path gains
decay faster than exponentially, then capacity is unbounded in the
transmit power.
\end{abstract}
\section{Introduction}
\label{sec:intro}
This paper studies the capacity of multipath (frequency-selective)
fading channels. A noncoherent channel model is considered where
neither transmitter nor receiver are cognizant of the fading's
realization, but both are aware of its statistic. Our focus is on the
high signal-to-noise ratio (SNR) regime.

For the special case of noncoherent \emph{frequency-flat} fading channels
(where we have only \emph{one} path), it was shown by Lapidoth \&
Moser \cite{lapidothmoser03_3} that if the fading process is of finite
entropy rate, then at high SNR capacity grows double-logarithmically
with the SNR. This is in stark contrast to the logarithmic growth of
the capacity of coherent fading channels (where the realization of the
fading is known to the receiver) \cite{ericson70}. Thus, communicating
over noncoherent flat-fading channels at high SNR is power
inefficient.

Recently, it has been demonstrated that communicating over
noncoherent \emph{multipath} fading channels at high SNR is not
merely power inefficient, but may be even worse: if the delay spread
is large in the sense that the variances of the path gains decay
\emph{exponentially or slower}, then capacity is \emph{bounded} in
the SNR; see \cite[Thm.~1]{kochlapidoth08_1}. For such channels, capacity does
not tend to infinity as the SNR tends to infinity.

In contrast, if the variances of the path gains decay \emph{faster
  than double-exponentially}, then capacity
is \emph{unbounded} in the SNR; see
\cite[Thm.~2]{kochlapidoth08_1}. This condition is certainly
satisfied if the number of paths is finite, i.e., if the channel
output is only influenced by the present and by the $\const{L}$
previous channel inputs. (Here only the variances of the
first $(\const{L}+1)$ path gains are positive, while the other
variances are zero.) It was shown in \cite{kochlapidoth08_2}
that in this case capacity is not only unbounded in the SNR, but
its growth with the $\SNR$ is also independent of the number of
paths and equals the growth of the capacity of
noncoherent frequency-flat fading channels, i.e.,
\begin{equation*}
  \lim_{\SNR\to\infty}\frac{C(\SNR)}{\log\log\SNR} = 1.
\end{equation*}
Thus, for finite $\const{L}$, the capacity pre-loglog is unaffected by
the number of paths $\const{L}$.

The above results demonstrate that whether the
capacity of a multipath channel is unbounded in the SNR
depends critically on the decay rate of the variances of the path
gains. However, \cite[Thm.~1]{kochlapidoth08_1} only accounts for
decay rates that are exponentially or slower, whereas
\cite[Thm.~2]{kochlapidoth08_1} only regards decay rates that are
faster than double-exponentially. Thus,
\cite[Thm.~1]{kochlapidoth08_1} \& \cite[Thm.~2]{kochlapidoth08_1}
fail to characterize the capacity of channels for which the variances
of the path gains decay faster than exponentially but slower than
double-exponentially. In this paper, we bridge this gap by showing
that if the variances of the path gains decay faster than
exponentially, then capacity is unbounded in the SNR.

\subsection{Channel Model}
\label{sub:channel}
Let $\Complex$ and $\Naturals$ denote the set of complex numbers and
the set of positive integers, respectively. We consider a
discrete-time multipath fading channel whose channel output
$Y_k\in\Complex$ at time $k\in\Naturals$ corresponding to the time-1
through time-$k$ channel inputs $x_1,\ldots,x_k\in\Complex$ is given
by
\begin{equation}
  Y_k = \sum_{\ell=0}^{k-1} H_k^{(\ell)} x_{k-\ell}+Z_k, \quad
  k\in\Naturals.
\end{equation}
Here $\{Z_k\}$ models additive noise, and $H_k^{(\ell)}$ denotes the
time-$k$ gain of the $\ell$-th path. We assume that $\{Z_k\}$ is
a sequence of independent and identically distributed (IID),
zero-mean, variance-$\sigma^2$, circularly-symmetric, complex Gaussian
random variables. For each path $\ell\in\Naturals_0$ (where
$\Naturals_0$ denotes the set of nonnegative integers), we assume that
$\bigl\{H_k^{(\ell)},\;k\in\Naturals\bigr\}$ is a zero-mean, complex stationary
process. We denote its variance and its differential entropy rate by
\begin{equation}
  \alpha_{\ell} \triangleq \E{\bigl|H_k^{(\ell)}\bigr|^2}, \qquad \ell\in\Naturals_0
\end{equation}
and
\begin{equation}
  h_{\ell} \triangleq \lim_{n\to\infty} \frac{1}{n}
  h\bigl(H_1^{(\ell)},\ldots,H_n^{(\ell)}\bigr), \quad \ell\in\Naturals_0.
\end{equation}
Without loss of generality we assume that $\alpha_0>0$. We further
assume that
\begin{equation}
  \label{eq:alpha}
  \sum_{\ell=0}^{\infty}\alpha_{\ell} \triangleq \alpha <\infty
\end{equation}
and
\begin{equation}
  \inf_{\ell\in\set{L}} h_{\ell} > -\infty,
\end{equation}
where the set $\set{L}$ is defined as $\set{L}\triangleq
\{\ell\in\Naturals_0:\alpha_{\ell}>0\}$. We finally assume that the
processes
\begin{equation*}
  \bigl\{H_k^{(0)},\;k\in\Naturals\bigr\},\bigl\{H_k^{(1)},\;k\in\Naturals\bigr\},\ldots
\end{equation*}
are independent (``uncorrelated scattering''); that they are jointly
independent of $\{Z_k\}$; and that the joint law of
\begin{equation*}
  \left(\{Z_k\},\bigl\{H_k^{(0)},\;k\in\Naturals\bigr\},\bigl\{H_k^{(1)},\;k\in\Naturals\bigr\},\ldots\right)
\end{equation*}
does not depend on the input sequence $\{x_k\}$. We consider a
noncoherent channel model where neither transmitter nor receiver is
cognizant of the realization of
$\bigl\{H_k^{(\ell)},\;k\in\Naturals\bigr\}$, $\ell\in\Naturals_0$,
but both are aware of their law. We do not assume that the path
gains are Gaussian.

\subsection{Channel Capacity}
\label{sub:capacity}
Let $A_m^n$ denote the sequence $A_m,\ldots,A_n$. We define the capacity as
\begin{equation}
  \label{eq:capacity}
  C(\SNR) \triangleq \varliminf_{n\to\infty}\frac{1}{n} \sup I\bigl(X_1^n;Y_1^n\bigr),
\end{equation}
where the supremum is over all joint distributions on $X_1,\ldots,X_n$
satisfying the power constraint
\begin{equation}
  \frac{1}{n} \sum_{k=1}^n \E{|X_k|^2} \leq \const{P},
\end{equation}
and where SNR is defined as
\begin{equation}
  \SNR \triangleq \frac{\const{P}}{\sigma^2}.
\end{equation}
By Fano's inequality, no rate above $C(\SNR)$ is achievable. (See
\cite{coverthomas91} for a definition of an achievable rate.) We do
not claim that there is a coding theorem associated with
\eqref{eq:capacity}, i.e., that $C(\SNR)$ is achievable. A coding
theorem will hold, for example, if there are only
$(\const{L}+1)$ paths (for some $\const{L}<\infty$), and if the processes
corresponding to these paths
\begin{equation*}
  \bigl\{H_k^{(0)},\;k\in\Naturals\bigr\},\ldots,\bigl\{H_k^{(\const{L})},\;k\in\Naturals\bigr\}
\end{equation*}
are jointly ergodic, see \cite{kim08}.

In \cite{kochlapidoth08_1} a necessary and a sufficient condition for
$C(\SNR)$ to be bounded in $\SNR$ was derived:
\begin{theorem}
  \label{thm:old}
  Consider the above channel model. Then
  \begin{equation}
    \left(\varliminf_{\ell\to\infty}
      \frac{\alpha_{\ell+1}}{\alpha_{\ell}}>0\right)
      \Longrightarrow
      \left(\sup_{\SNR>0} C(\SNR) < \infty\right)\label{eq:old_first}
    \end{equation}
    and
    \begin{equation}
      \left(\lim_{\ell\to\infty}
        \frac{1}{\ell}\log\log\frac{1}{\alpha_{\ell}}=\infty\right) 
    \Longrightarrow \left(\sup_{\SNR>0} C(\SNR) = \infty\right),\label{eq:old_second}
  \end{equation}
  where we define $a/0\triangleq\infty$ for every $a>0$ and $0/0\triangleq
  0$.
\end{theorem}
\begin{proof}
  For the first condition \eqref{eq:old_first} see
  \cite[Thm.~1]{kochlapidoth08_1}, and for the second condition \eqref{eq:old_second} see
  \cite[Thm.~2]{kochlapidoth08_1}.
\end{proof}
For example, when $\alpha_{\ell}=e^{-\ell}$, then capacity is
bounded, and when $\alpha_{\ell}=\exp\bigl(-\exp(\ell^{\kappa})\bigr)$ for
some $\kappa>1$, then capacity is unbounded.
Roughly speaking, we can say that
when $\{\alpha_{\ell}\}$ decays exponentially or slower, then
$C(\SNR)$ is bounded in $\SNR$, and when $\{\alpha_{\ell}\}$ decays
faster than double-exponentially, then $C(\SNR)$ is unbounded
in $\SNR$.

\subsection{Main Result}
\label{sub:result}
Our main result is an improved achievability result. We derive a
weaker condition that satisfies to guarantee that capacity is
unbounded in the $\SNR$.

\begin{theorem}
  \label{thm:main}
  Consider the above channel model. Then
  \begin{equation}
    \left(\lim_{\ell\to\infty}\frac{1}{\ell}\log\frac{1}{\alpha_{\ell}}
    =  \infty\right)
    \Longrightarrow \left(\sup_{\SNR>0} C(\SNR)=\infty\right),
  \end{equation}
  where we define $1/0\triangleq\infty$.
\end{theorem}
\begin{proof}
  See Section~\ref{sec:proof}.
\end{proof}
By noting that
\begin{equation*}
  \left(\lim_{\ell\to\infty}\frac{\alpha_{\ell+1}}{\alpha_{\ell}}=0\right)
  \Longrightarrow \left(\lim_{\ell\to\infty}\frac{1}{\ell}\log\frac{1}{\alpha_{\ell}}=\infty\right)
\end{equation*}
we obtain from Theorems \ref{thm:old} \& \ref{thm:main} the immediate corollary:
\begin{corollary}
  \label{cor:main}
  Consider the above channel model. Then
  \begin{IEEEeqnarray}{rrcl}
    i) & \left(\varliminf_{\ell\to\infty}
      \frac{\alpha_{\ell+1}}{\alpha_{\ell}}>0\right)&
    \;\Longrightarrow\;  & \left(\sup_{\SNR>0} C(\SNR) <
      \infty\right)\\
    ii) & \;
    \left(\lim_{\ell\to\infty}\frac{\alpha_{\ell+1}}{\alpha_{\ell}}=0\right) &
    \;\Longrightarrow \; &\left(\sup_{\SNR>0} C(\SNR)=\infty\right),\IEEEeqnarraynumspace
  \end{IEEEeqnarray}
  where we define $a/0\triangleq\infty$ for every $a>0$ and $0/0\triangleq
  0$.
\end{corollary}
For example, when $\alpha_{\ell}=\exp(-\ell^{\kappa})$ for some
$\kappa>1$, then capacity is unbounded. 

Theorem~\ref{thm:main} and Corollary~\ref{cor:main} demonstrate
that when $\{\alpha_{\ell}\}$ decays faster than exponentially,
then $C(\SNR)$ is unbounded in $\SNR$, thus bridging the gap between
\eqref{eq:old_first} and \eqref{eq:old_second}.

\section{Proof of Theorem 2}
\label{sec:proof}
In order to prove Theorem~\ref{thm:main}, we shall derive in
Section~\ref{sub:lowerbound} a lower bound on capacity and then show in
Section~\ref{sub:unbounded} that this bound can be made arbitrarily
large, provided that
\begin{equation*}
  \lim_{\ell\to\infty}\frac{1}{\ell} \log\frac{1}{\alpha_{\ell}} = \infty.
\end{equation*}

\subsection{Capacity Lower Bound}
\label{sub:lowerbound}
To derive a lower bound on capacity, we evaluate
$\frac{1}{n}I(X_1^n;Y_1^n)$ for the following distribution on the
inputs $\{X_k\}$.

Let $\const{L}(\const{P})$ be such that
\begin{equation}
  \label{eq:L}
  \sum_{\ell=\const{L}(\const{P})+1}^\infty \alpha_{\ell}\cdot\const{P} \leq \sigma^2.
\end{equation}
To shorten notation, we shall write in the following $\const{L}$ instead
of $\const{L}(\const{P})$. Let $\tau\in\Naturals$ be some positive
integer that possibly depends on $\const{L}$, and let
$\vect{X}_b=(X_{b(\const{L}+\tau)+1},\ldots,X_{(b+1)(\const{L}+\tau)})$.
We choose $\{\vect{X}_b\}$ to be IID with
\begin{equation*}
  \vect{X}_b = \bigl(\underbrace{0,\ldots,0}_{\const{L}},\tilde{X}_{b\tau+1},\ldots,\tilde{X}_{(b+1)\tau}\bigr),
\end{equation*}
where $\tilde{X}_{b\tau+1},\ldots,\tilde{X}_{(b+1)\tau}$ is a sequence
of independent, zero-mean, circularly-symmetric, complex random variables with
$\log|\tilde{X}_{b\tau+\nu}|^2$ being uniformly distributed over the
interval $\bigl[\log \const{P}^{(\nu-1)/\tau},\log
\const{P}^{\nu/\tau}\bigr]$, i.e., for each $\nu=1,\ldots,\tau$
\begin{equation*}
  \log |\tilde{X}_{b\tau+\nu}|^2 \sim
  \Uniform{\bigl[\log\const{P}^{(\nu-1)/\tau},\log
  \const{P}^{\nu/\tau}\bigr]}.
\end{equation*}
(Here and throughout this proof we assume that $\const{P}>1$.)

Let $\kappa\triangleq \lfloor \frac{n}{\const{L}+\tau}\rfloor$ (where
$\lfloor a \rfloor$ denotes the largest integer that is less
than or equal to $a$), and
let $\vect{Y}_b$ denote the vector
$(Y_{b(L+\tau)+1},\ldots,Y_{(b+1)(\const{L}+\tau)})$. By the chain
rule for mutual information \cite[Thm.~2.5.2]{coverthomas91} we have
\begin{IEEEeqnarray}{lCl}
  I\bigl(X_1^n;Y_1^n\bigr) & \geq &
  I\bigl(\vect{X}_0^{\kappa-1};\vect{Y}_0^{\kappa-1}\bigr)\nonumber\\
  & = & \sum_{b=0}^{\kappa-1}
  I\bigl(\vect{X}_b;\vect{Y}_0^{\kappa-1}\bigm|\vect{X}_0^{b-1}\bigr)\nonumber\\
  & \geq & \sum_{b=0}^{\kappa-1} I(\vect{X}_b;\vect{Y}_b),\label{eq:proof1}
\end{IEEEeqnarray}
where the first inequality follows by restricting the number of
observables; and where the last inequality follows by restricting the
number of observables and by noting that $\{\vect{X}_b\}$ is IID.

We continue by lower bounding each summand on the right-hand side
(RHS) of \eqref{eq:proof1}. We use again the chain rule and that
reducing observations cannot increase mutual information to obtain
\begin{IEEEeqnarray}{lCl}
  I(\vect{X}_b;\vect{Y}_b) & = & \sum_{\nu=1}^{\tau}
  I\bigl(\tilde{X}_{b\tau+\nu};\vect{Y}_b\bigm|\tilde{X}_{b\tau+1}^{b\tau+\nu-1}\bigr)\nonumber\\
  & \geq & \sum_{\nu=1}^{\tau} I\bigl(\tilde{X}_{b\tau+\nu};Y_{b(\const{L}+\tau)+\const{L}+\nu}\bigm|\tilde{X}_{b\tau+1}^{b\tau+\nu-1}\bigr)\nonumber\\
  & \geq & \sum_{\nu=1}^{\tau} I\bigl(\tilde{X}_{b\tau+\nu};Y_{b(\const{L}+\tau)+\const{L}+\nu}\bigr),\label{eq:proof1b}
\end{IEEEeqnarray}
where we have additionally used in the last inequality that
$\tilde{X}_{b\tau+1},\ldots,\tilde{X}_{(b+1)\tau}$ are independent.

Defining
\begin{IEEEeqnarray}{lCl}
  W_{b\tau+\nu} & \triangleq &
  \sum_{\ell=1}^{b(\const{L}+\tau)+\const{L}+\nu-1}
  H_{b(\const{L}+\tau)+\const{L}+\nu}^{(\ell)}
  X_{b(\const{L}+\tau)+\const{L}+\nu-\ell}\nonumber\\
  & & \qquad \qquad \qquad {} + Z_{b(\const{L}+\tau)+\const{L}+\nu}
\end{IEEEeqnarray}
each summand on the RHS of \eqref{eq:proof1b} can be written as
\begin{IEEEeqnarray}{lCl}
  \IEEEeqnarraymulticol{3}{l}{I\bigl(\tilde{X}_{b\tau+\nu};Y_{b(\const{L}+\tau)+\const{L}+\nu}\bigr)}\nonumber\\
  \qquad \quad & = & I\bigl(\tilde{X}_{b\tau+\nu};H_{b(\const{L}+\tau)+\const{L}+\nu}^{(0)}\tilde{X}_{b\tau+\nu}+W_{b\tau+\nu}\bigr).\IEEEeqnarraynumspace\label{eq:proof2}
\end{IEEEeqnarray}
A lower bound on \eqref{eq:proof2} follows from the following lemma.
\begin{lemma}
  \label{lemma}
  Let the random variables $X$, $H$, and $W$ have finite second
  moments. Assume that both $X$ and $H$ are of finite differential
  entropy. Finally, assume that $X$ is independent of $H$; that
  $X$ is independent of $W$; and that $X \markov H \markov W$ forms a
  Markov chain. Then
  \begin{IEEEeqnarray}{lCl}
    I(X;HX+W) & \geq & h(X)-\E{\log|X|^2}+\E{\log|H|^2}\nonumber\\
    & & {} - \E{\log\biggl(\pi e\biggl(\sigma_H+\frac{\sigma_W}{|X|}\biggr)^2\biggr)},\IEEEeqnarraynumspace
  \end{IEEEeqnarray}
  where $\sigma^2_H\geq 0$ and $\sigma^2_H>0$ denote the variances of\/
  $W$ and $H$. (Note that the assumptions that $X$ and $H$ have finite
  second moments and are of finite differential entropy guarantee
  that $\E{\log|X|^2}$ and $\E{\log|H|^2}$ are finite, see \textnormal{\cite[Lemma
  6.7e]{lapidothmoser03_3}}.)
\end{lemma}
\begin{proof}
  See \cite[Lemma 4]{lapidoth05}.
\end{proof}
It can be easily verified that for the channel model given in
Section~\ref{sub:channel} and for the above coding scheme the lemma's
conditions are satisfied. We therefore obtain from Lemma~\ref{lemma}
\begin{IEEEeqnarray}{lCl}
  \IEEEeqnarraymulticol{3}{l}{I\bigl(\tilde{X}_{b\tau+\nu};H_{b(\const{L}+\tau)+\const{L}+\nu}^{(0)}\tilde{X}_{b\tau+\nu}+W_{b\tau+\nu}\bigr)}\nonumber\\
  & \geq & h\bigl(\tilde{X}_{b\tau+\nu}\bigr)
  -\E{\log|\tilde{X}_{b\tau+\nu}|^2} +
  \E{\log\bigl|H_{b(\const{L}+\tau)+\const{L}+\nu}^{(0)}\bigr|^2}\nonumber\\
  & & {} - \E{\log\biggl(\pi e\biggl(\sqrt{\alpha_0}+\frac{\sqrt{\E{|W_{b\tau+\nu}|^2}}}{|\tilde{X}_{b\tau+\nu}|}\biggr)^2\biggr)}.\label{eq:proof3}
\end{IEEEeqnarray}
Using that the differential entropy of a circularly-symmetric random
variable is given by (see \cite[Eqs.~(320) \&
(316)]{lapidothmoser03_3})
\begin{equation}
  h\bigl(\tilde{X}_{b\tau+\nu}\bigr) =
  \E{\log|\tilde{X}_{b\tau+\nu}|^2}+h\bigl(\log|\tilde{X}_{b\tau+\nu}|^2\bigr)
  + \log\pi,
\end{equation}
and evaluating $h(\log|\tilde{X}_{b\tau+\nu}|^2)$ for our choice of
$\tilde{X}_{b\tau+\nu}$, yields for the first two terms on the RHS of
\eqref{eq:proof3}
\begin{equation}
  \label{eq:proof3b}
  h\bigl(\tilde{X}_{b\tau+\nu}\bigr)
  -\E{\log|\tilde{X}_{b\tau+\nu}|^2} = \log\log\const{P}^{1/\tau} +\log\pi.
\end{equation}
We next upper bound
\begin{IEEEeqnarray}{l}
  \frac{\E{|W_{b\tau+\nu}|^2}}{|\tilde{X}_{b\tau+\nu}|^2} =
  \sum_{\ell=1}^{\const{L}} \alpha_{\ell}
  \frac{\E{|X_{b(\const{L}+\tau)+\const{L}+\nu-\ell}|^2}}{|\tilde{X}_{b\tau+\nu}|^2}\nonumber\\
  {} + \!\sum_{\ell=\const{L}+1}^{b(\const{L}+\tau)+\const{L}+\nu-1}\! \alpha_{\ell}
  \frac{\E{|X_{b(\const{L}+\tau)+\const{L}+\nu-\ell}|^2}}{|\tilde{X}_{b\tau+\nu}|^2}
  +\frac{\sigma^2}{|\tilde{X}_{b\tau+\nu}|^2}.\IEEEeqnarraynumspace\label{eq:proof4}
\end{IEEEeqnarray}
To this end we note that for our choice of $\{X_k\}$ and by the
assumption that $\const{P}>1$, we have
\begin{equation}
  \E{|X_{\ell}|^2} \leq
  \const{P},\quad \ell\in\Naturals,
\end{equation}
\begin{equation}
\E{|X_{b(\const{L}+\tau)+\const{L}+\nu-\ell}|^2}\leq
\const{P}^{(\nu-\ell)/\tau}, \quad \ell=1,\ldots,\const{L},
\end{equation}
 and
\begin{equation}
  \label{eq:proof4a}
  |\tilde{X}_{b\tau+\nu}|^2\geq\const{P}^{(\nu-1)/\tau}\geq 1,
\end{equation}
from which we obtain
\begin{equation}
  \label{eq:proof4b}
  \frac{\E{|X_{b(\const{L}+\tau)+\const{L}+\nu-\ell}|^2}}{|\tilde{X}_{b\tau+\nu}|^2}
  \leq
  \frac{\const{P}^{(\nu-\ell)/\tau}}{\const{P}^{(\nu-1)/\tau}}\leq 1,
  \quad \ell=1,\ldots,\const{L}
\end{equation}
and
\begin{equation}
  \label{eq:proof4c}
  \frac{\E{|X_{b(\const{L}+\tau)+\const{L}+\nu-\ell}|^2}}{|\tilde{X}_{b\tau+\nu}|^2}
  \leq \const{P}, \quad \const{L}<\ell<b(\const{L}+\tau)+\const{L}+\nu.
\end{equation}
Applying \eqref{eq:proof4a}--\eqref{eq:proof4c} to \eqref{eq:proof4}
yields
\begin{IEEEeqnarray}{lCl}
  \frac{\E{|W_{b\tau+\nu}|^2}}{|\tilde{X}_{b\tau+\nu}|^2}
  & \leq & \sum_{\ell=1}^{\const{L}}\alpha_{\ell} +
  \sum_{\ell=\const{L}+1}^{b(\const{L}+\tau)+\const{L}+\nu-1}
  \alpha_{\ell} \cdot \const{P}+\sigma^2\nonumber\\
  & \leq & \alpha +
  \sum_{\ell=\const{L}+1}^{\infty}
  \alpha_{\ell} \cdot \const{P} + \sigma^2 \nonumber\\
  & \leq & \alpha + 2\sigma^2,\label{eq:proof5}
\end{IEEEeqnarray}
with $\alpha$ being defined in \eqref{eq:alpha}. Here
the second inequality follows because $\alpha_{\ell}$,
$\ell\in\Naturals_0$ and $\const{P}$
are nonnegative, and the last inequality follows from
\eqref{eq:L}.

By combining \eqref{eq:proof3} with \eqref{eq:proof3b} \&
\eqref{eq:proof5}, and by noting that by the stationarity of
$\bigl\{H_k^{(0)},\;k\in\Naturals\bigr\}$
\begin{equation*}
  \E{\log\bigl|H_{b(\const{L}+\tau)+\const{L}+\nu}^{(0)}\bigr|^2} =
  \E{\log\bigl|H_{1}^{(0)}\bigr|^2},
\end{equation*}
we obtain the lower bound
\begin{IEEEeqnarray}{lCl}
  \IEEEeqnarraymulticol{3}{l}{I\bigl(\tilde{X}_{b\tau+\nu};H_{b(\const{L}+\tau)+\const{L}+\nu}^{(0)}\tilde{X}_{b\tau+\nu}+W_{b\tau+\nu}\bigr)}\nonumber\\
  \qquad & \geq & \log\log\const{P}^{1/\tau} +
  \E{\log\bigl|H_{1}^{(0)}\bigr|^2} - 1\nonumber\\
  & & {} - 2\log\bigl(\sqrt{\alpha_0}+\sqrt{\alpha + 2\sigma^2}\bigr).\IEEEeqnarraynumspace\label{eq:proof6}
\end{IEEEeqnarray}
Note that the RHS of \eqref{eq:proof6} neither depends on $\nu$ nor on
$b$. We therefore have from \eqref{eq:proof6}, \eqref{eq:proof1b}, and
\eqref{eq:proof1}
\begin{equation}
  I\bigl(X_1^n;Y_1^n\bigr) \geq \kappa\tau \log\log\const{P}^{1/\tau}
  + \kappa\tau \Upsilon,\label{eq:proof7}
\end{equation}
where we define $\Upsilon$ as
\begin{equation}
  \Upsilon \triangleq \E{\log\bigl|H_{1}^{(0)}\bigr|^2} - 1 - 2\log\bigl(\sqrt{\alpha_0}+\sqrt{\alpha + 2\sigma^2}\bigr).
\end{equation}
Dividing the RHS of \eqref{eq:proof7} by $n$, and computing the limit
as $n$ tends to infinity, yields the lower bound on capacity
\begin{equation}
  \label{eq:lowerbound}
  C(\SNR) \geq \frac{\tau}{\const{L}+\tau}
  \log\log\const{P}^{1/\tau}+\frac{\tau}{\const{L}+\tau}\Upsilon,
  \quad \const{P}>1,
\end{equation}
where we have used that $\lim_{n\to\infty}\kappa/n=1/(\const{L}+\tau)$.

\subsection{Unbounded Capacity}
\label{sub:unbounded}
We next show that
\begin{equation}
  \label{eq:exponential}
  \lim_{\ell\to\infty}\frac{1}{\ell}\log\frac{1}{\alpha_{\ell}} = \infty
\end{equation}
implies that the RHS of \eqref{eq:lowerbound} can be made arbitrarily
large. To this end we note that by \eqref{eq:exponential} we can
find for every $0<\varrho<1$ an $\ell_0\in\Naturals$ such that
\begin{equation}
  \alpha_{\ell} < \varrho^{\ell}, \quad \ell > \ell_0.
\end{equation}
We therefore have
\begin{equation}
  \label{eq:sum}
  \sum_{\ell=\ell^{\prime}+1}^{\infty} \alpha_{\ell} <
  \sum_{\ell=\ell^{\prime}+1}^{\infty} \varrho^{\ell} =
  \varrho^{\ell^{\prime}}\frac{\varrho}{1-\varrho}, \quad
  \ell^{\prime}\geq\ell_0.
\end{equation}
We choose $\const{L}$ so that it satisfies
\begin{equation}
  \label{eq:otherL}
  \varrho^{\const{L}}\frac{\varrho}{1-\varrho}\const{P} \leq \sigma^2,
\end{equation}
i.e., we choose
\begin{equation}
  \label{eq:choice}
  \const{L} = \left\lceil\frac{\log\Bigl(\SNR\frac{\varrho}{1-\varrho}\Bigr)}{\log\frac{1}{\varrho}}\right\rceil
\end{equation}
(where $\lceil a \rceil$ denotes the smallest integer that is greater
than or equal to $a$). We shall argue next that this choice also satisfies
$\eqref{eq:L}$. Indeed, we have by \eqref{eq:choice} that $\const{L}$
tends to infinity as $\SNR\to\infty$, which implies that, for
sufficiently large $\SNR$, $\const{L}$ is greater than $\ell_0$. It
follows then from \eqref{eq:sum} and \eqref{eq:otherL} that
\begin{equation}
  \sum_{\ell=\const{L}+1}^{\infty} \alpha_{\ell}\cdot\const{P} <
  \varrho^{\const{L}}\frac{\varrho}{1-\varrho}\const{P} \leq \sigma^2.
\end{equation}

We continue by evaluating the RHS of \eqref{eq:lowerbound} for our
choice of $\const{L}$ \eqref{eq:choice} and for $\tau=\const{L}$
\begin{IEEEeqnarray}{lCl}
  C(\SNR) & \geq & \frac{\tau}{\const{L}+\tau}
  \log\log\const{P}^{1/\tau}+\frac{\tau}{\const{L}+\tau}\Upsilon\nonumber\\
  & = & \frac{1}{2} \log\biggl(\frac{\log\const{P}}{\const{L}}\biggr)+\frac{1}{2}\Upsilon.
\end{IEEEeqnarray}
Taking the limit as $\SNR$ tends to infinity yields
\begin{IEEEeqnarray}{lCl}
  \IEEEeqnarraymulticol{3}{l}{\lim_{\SNR\to\infty} C(\SNR)}\nonumber\\
  \qquad & \geq & \lim_{\SNR\to\infty}
  \frac{1}{2}
  \log\Biggl(\frac{\log(\SNR\cdot\sigma^2)}{\frac{\log(\SNR\cdot\varrho/(1-\varrho))}{\log(1/\varrho)}}\Biggr)
  +\frac{1}{2}\Upsilon\nonumber\\
  & = & \frac{1}{2}\log\log\frac{1}{\varrho} + \frac{1}{2}\Upsilon.
\end{IEEEeqnarray}
As this can be made arbitrarily large by choosing $\varrho$ sufficiently
small, we conclude that
\begin{equation*}
  \lim_{\ell\to\infty}\frac{1}{\ell}\log\frac{1}{\alpha_{\ell}} = \infty
\end{equation*}
implies that $C(\SNR)$ is unbounded in SNR.

\section{Summary}
\label{sec:summary}
We studied the capacity of discrete-time, noncoherent, multipath fading
channels. It was shown that if the variances of the path gains decay
faster than exponentially, then capacity is unbounded in the
SNR. This complements previous results obtained in
\cite{kochlapidoth08_1} and \cite{kochlapidoth08_2}.

The overall picture looks as follows: 
\begin{itemize}
\item If the number of paths is
infinite in the sense that the channel output is influenced by the
present and by \emph{all} previous channel inputs,  and if the
variances of the path gains decay exponentially or slower, then
capacity is bounded even as the SNR grows without bound.
\item If the number of paths is infinite but the variances of the
path gains decay faster than exponentially, then capacity tends to
infinity as $\SNR\to\infty$.
\item If the number of paths is finite, then, irrespective of the
  number of paths, the capacity pre-loglog is $1$. Thus, in this case
  the multipath behavior has no significant effect on the high-SNR
  capacity.
\end{itemize}
\vfill\pagebreak

We thus see that the high-SNR behavior of the capacity of noncoherent
multipath fading channels depends critically on the assumed channel
model. Consequently, when studying such channels at high SNR, the
channel modeling is crucial, as slight changes
in the model might lead to completely different capacity results.

\section{Acknowledgment}
The authors wish to thank Olivier Leveque and Nihar Jindal for their
comments which were the inspiration for the proof of
Theorem~\ref{thm:main}.

\end{document}